\documentclass{llncs}

\usepackage{amsmath,amssymb,framed,enumerate,fullpage}
\usepackage{algorithm}
\usepackage{algorithmic,eqparbox,array}

\newcommand{\comment}[1]{}
\newcommand{\ltsqn}[1]{\|#1\|^2}
\newcommand{\ltn}[1]{\|#1\|}

\newcommand{\bB}{{\bf B}}

\newcommand{\bS}{{\bf S}}
\newcommand{\bD}{{\bf D}}

\newcommand{\bp}{{\bf p}}
\newcommand{\bb}{{\bf b}}
\newcommand{\bt}{{\bf t}}
\newcommand{\bz}{{\bf z}}
\newcommand{\bv}{{\bf v}}
\newcommand{\bu}{{\bf u}}
\newcommand{\bx}{{\bf x}}

\newcommand{\R}{{\bf\mathbb{R}}}
\newcommand{\Z}{{\bf\mathbb{Z}}}
\newcommand{\btd}{{\bf\texttt{d}}}
\newcommand{\bL}{{\bf\mathbb{L}}}
\newcommand{\Q}{{\bf\mathbb{Q}}}

\newcommand{\tbL}{\tilde{\bL}}

\newcommand{\tgsvp}{{\tt GapSVP}}
\newcommand{\tgcvp}{{\tt GapCVP}}
\newcommand{\tcvp}{{\tt CVP}}

\newcommand{\tsat}{{\tt SAT}}
\newcommand{\tsvp}{{\tt SVP}}
\newcommand{\tusvp}{{\tt uSVP}}
\newcommand{\tdusvp}{{\tt duSVP}}
\newcommand{\tcoam}{\text{co-}\textbf{AM}}
\newcommand{\tconp}{\text{co-}\textbf{NP}}
\newcommand{\tnp}{\textbf{NP}}
\newcommand{\tpoly}{{\tt poly}}

\title{Improved hardness results for unique shortest vector problem}

\author{Divesh Aggarwal \and Chandan Dubey}
\institute{Institute for Theoretical Computer Science\\ETH Zurich\\\email{divesha@inf.ethz.ch}\\\email{chandan.dubey@inf.ethz.ch}}

\comment{
\author{Divesh Aggarwal \and Chandan Dubey\thanks{Chandan Dubey is partially supported by the Swiss National Science Foundation (SNF), project no. 200021-132508}}
\institute{Institute for Theoretical Computer Science\\ETH Zurich\\\email{divesha@inf.ethz.ch}\\\email{chandan.dubey@inf.ethz.ch}}
}

\begin{document}

\maketitle

\begin{abstract}
We give several improvements on the known hardness of the unique shortest vector problem.
\begin{itemize}
\item We give a deterministic reduction from the shortest vector problem to the unique shortest vector problem. As a byproduct, we get deterministic NP-hardness for unique shortest vector problem in the $\ell_\infty$ norm. 
\item We give a randomized reduction from $\tsat$ to $\tusvp_{1+1/\tpoly(n)}$. This shows that $\tusvp_{1+1/\tpoly(n)}$ is NP-hard under randomized reductions.
\item We show that if $\tgsvp_{\gamma} \in \tconp$ (or $\tcoam$) then $\tusvp_{\sqrt{\gamma}} \in \tconp$ ($\tcoam$ respectively). This simplifies previously known $\tusvp_{n^{1/4}} \in \tcoam$ proof by Cai \cite{Cai98} to $\tusvp_{(n/\log n)^{1/4}} \in \tcoam$, and additionally generalizes it to $\tusvp_{n^{1/4}} \in \tconp$.
\item We give a deterministic reduction from search-$\tusvp_{\gamma}$ to the decision-$\tusvp_{\gamma/2}$. We also show that the decision-$\tusvp$ is {\bf NP}-hard for randomized reductions, which does not follow from Kumar-Sivakumar \cite{KS01}.
\end{itemize}
\end{abstract}
\newpage

\section{Introduction}

A {\em lattice} is the set of all integer combinations of~$n$ linearly independent vectors $\bb_1, \bb_2, \dots, \bb_n$ in $\R^m$. These vectors are referred to as a {\em basis} of the lattice and $n$ is the {\em rank} of the lattice. The {\em successive minima} $\lambda_i(\bL)$ (where $i = 1, \dots, n$) for the lattice $\bL$ are among the most fundamental parameters associated to a lattice. The $\lambda_i(\bL)$ is defined as the smallest value such that a sphere of radius $\lambda_i(\bL)$ centered around the origin contains at least $i$ linearly independent lattice vectors. Lattices have been investigated by computer scientists for a few decades after the discovery of the LLL algorithm \cite{LLL82}. More recently, Ajtai \cite{Ajt96} showed that lattice problems have a very desirable property for cryptography i.e., they exhibit a worst-case to average-case reduction. This property immediately yields one-way functions and collision resistant hash functions, based on the {\em worst case} hardness of lattice problems. This is in a stark contrast to the traditional number theoretic constructions which are based on the average-case hardness e.g., factoring, discrete logarithms. 

We now describe some of the most fundamental and widely studied lattice problems. Given a lattice $\bL$, the $\gamma$-approximate shortest vector problem ($\tsvp_\gamma$) is the problem of finding a non-zero lattice vector of length at most $\gamma\lambda_1(\bL)$. Let the minimum distance of a point $\bt \in \R^m$ from a vector of the lattice $\bL$ be denoted by $\btd(\bt,\bL)$. Given a lattice $\bL$ and a point $\bt\in \R^m$, the $\gamma$-approximate closest vector problem or $\tcvp_\gamma$, is the problem of finding a $\bv \in \bL$ such that $\ltn{\bv-\bt} \leq \gamma\btd(\bt,\bL)$. 

Besides the search version just described, $\tcvp$ and $\tsvp$ also have a decision version. The problem $\tgcvp_\gamma$ is the problem of deciding if, given $(\bB,\bt, d \in \R)$, $\btd(\bt,\bL(\bB))\leq d$ or $\btd(\bt,\bL(\bB)) > \gamma d$. Similarly, the problem $\tgsvp_\gamma$  is the problem of deciding if, given $(\bB,d\in \R)$, $\lambda_1(\bL(\bB)) \leq d$ or $\lambda_1(\bL(\bB)) > \gamma d$.

The two problems $\tcvp$ and $\tsvp$ are quite well studied. We know that  they can be solved exactly in deterministic $2^{O(n)}$ time \cite{MV10,AKS01}. They can be approximated within a factor of $2^{n(\log\log n)^2/\log n}$, in polynomial time, using LLL \cite{LLL82} and subsequent improvements by Schnorr \cite{Sch87} (for details, see the book by Micciancio and Goldwasser \cite{GM02}). On the other hand, it is known that there exists $c>0$, such that no polynomial time algorithm can approximate these problems within a factor of $n^{c/\log\log n}$, unless {\bf P} $=$ {\bf NP} or another unlikely scenario is true \cite{DKRS03,HR07,BS99}. It is also known that both these problems cannot be NP-hard for a factor of $\sqrt{n/\log n}$ or the polynomial hierarchy will collapse.

A variant of $\tsvp$ that has been especially relevant in cryptography is the unique shortest vector problem ($\tusvp$). The problem $\tusvp_{\gamma}$ is the problem of finding the shortest non-zero vector of the lattice, given the promise that $\lambda_2(\bL) \geq \gamma\lambda_1(\bL)$. The security of the first public key cryptosystem by Ajtai-Dwork \cite{AD97} was based on the worst-case hardness of $\tusvp_{O(n^8)}$. In a series of papers \cite{GGH97,Reg04}, the uniqueness factor was reduced to $O(n^{1.5})$. 

In contrast to $\tcvp$ and $\tsvp$, much less is known about the hardness of $\tusvp$. The current {\bf NP}-hardness result known for $\tusvp_{\gamma}$ is for $\gamma < {1+2^{-n^c}}$, which is shown by a randomized reduction from $\tsvp$ \cite{KS01}. In \cite{LM09}, it was shown that there is a reduction from $\tusvp_\gamma$ to $\tgsvp_\gamma$ and also a reduction from $\tgsvp_{\gamma}$ to $\tusvp_{\frac{\gamma}{2\sqrt{n/\log n}}}$. From the first reduction, we can conclude that $\tusvp_{\gamma} \in \tconp$ if $\tgsvp_{\gamma} \in \tconp$ which, using the result of \cite{AR05} implies that $\tusvp_{\sqrt{n}} \in \tconp$. It is already know from Cai \cite{Cai98} that $\tusvp_{n^{1/4}} \in \tcoam$. A discussion of the proofs and the simplification can be found in Section~5.

{\bf Contributions of this paper.} In Section 3.1, we give a deterministic polynomial time reduction from $\tsvp$ to $\tusvp$ achieving similar bounds as \cite{KS01} for the $\ell_2$ norm. This implies, unlike \cite{KS01}, that deterministic {\bf NP}-hardness of $\tsvp$ implies deterministic {\bf NP}-hardness of $\tusvp$. Also, this result shows that the decision problem $\tdusvp$ is also {\bf NP}-hard under randomized reductions. In Section 3.2, we show that a similar idea gets us NP-hardness proof for $\tusvp$ in $\ell_\infty$ norm. In Section 4, we show that $\tusvp_{1+1/{\texttt poly}(n)}$ is hard by giving a randomized reduction of the $\tsvp$ instance created by Khot \cite{Khot05} to $\tusvp_{1+1/{\texttt poly}(n)}$. In Section 5, we show  $\tusvp_{c(n)^{1/4}} \in \tconp$ for some $c>0$, which implies that $\tusvp_{\gamma}$ cannot be {\bf NP}-hard for $\gamma \geq c n^{1/4}$ unless ${\bf NP}=\tconp$. In Section 6, we give a search to decision reduction for the unique shortest vector problem, i.e., a reduction from $\tusvp_\gamma$ to $\tdusvp_{\gamma/2}$. The definition of $\tdusvp$ is implicit in Cai \cite{Cai98}. A comparison of some of our results with previously known results has been depicted in Figures \ref{fig:now} and \ref{fig:prev}.

%

\begin{figure}[htb]
\begin{picture}(100,20)(-50,-10)
\put(0,0){\line(1,0){300}}
\put(270,-5){\line(0,1){10}}
\put(265,-10){\texttt{crypto}}
\put(265,10){$n^{1.5}$}
\put(220,-5){\line(0,1){10}}
\put(218,10){\texttt{?}}
\put(215,-10){$\tconp$}
\put(190,-5){\line(0,1){10}}
\put(180,-10){$\tcoam$}
\put(180,10){$n^{\frac{1}{4}}$}
\put(32,-5){\line(0,1){10}}
\put(17,-12){\texttt{NP-hard}}
\put(7,10){$(1+\frac{1}{\texttt{exp}})$}
\end{picture}
\caption{Before this paper}
\label{fig:now}
\end{figure}

\begin{figure}[htb]
\begin{picture}(100,20)(-50,0)
\put(0,0){\line(1,0){300}}
\put(270,-5){\line(0,1){10}}
\put(265,-10){\texttt{crypto}}
\put(265,10){$n^{1.5}$}
\put(190,-5){\line(0,1){10}}
\put(184,10){$n^{\frac{1}{4}}$}
\put(180,-10){$\tconp$}
\put(160,-5){\line(0,1){10}}
\put(145,-10){$\tcoam$}
\put(143,10){$(\frac{n}{\log n})^{\frac{1}{4}}$}
\put(40,-5){\line(0,1){10}}
\put(30,-12){\texttt{NP-hard}}
\put(30,10){$(1+\frac{1}{\texttt{poly}})$}
\end{picture}
\caption{After this paper}
\label{fig:prev}
\end{figure}

\section{Preliminaries}

\subsection{Notation}
A lattice basis is a set of linearly independent vectors $\bb_1, \dots, \bb_n \in \R^m$.  It is sometimes convenient to think of the basis as an $m \times n$ matrix $\bB$, whose $n$ columns are the vectors $\bb_1, \dots, \bb_n$. The lattice generated by the basis $\bB$ will be written as $\bL(\bB)$ and is defined as $\bL(\bB) = \{\bB \bx | \bx \in \Z^n\}$. A vector $\bv \in \bL$ is called a primitive vector of the lattice $\bL$ if it is not an integer multiple of another lattice vector except $\pm \bv$. We will assume that the lattice is over rationals, i.e., $\bb_1, \dots, \bb_n \in \Q^m$, and the entries are represented by the pair of numerator and denominator. 

A {\em shortest vector} of a lattice is a non-zero vector in the lattice whose $\ell_2$ norm is minimal.   The length of the shortest vector is $\lambda_1(\bL(\bB))$, where $\lambda_1$ is as defined in the introduction.  For a vector $\bt \in \R^m$, let $\btd(\bt,\bL(\bB))$ denote the distance of $\bt$ to the closest lattice point in $\bL(\bB)$. 

For any lattice $\bL$, and any vector $\bv \in \bL$, we denote by $\bL_{\bot \bv}$ the lattice obtained by projecting $\bL$ to the space orthogonal to $\bv$. 

For an integer $k \in \Z^{+}$ we use $[k]$ to denote the set $\{1, \ldots, k\}$. 

\subsection{Lattice Problems}

In this paper we are concerned with the shortest vector problem and the unique shortest vector problem.  The search and decision versions  of the shortest vector problem are defined below. 

\begin{description}
\item[$\tgsvp_{\gamma}$:] Given a lattice basis $\bB$ and an integer $d$, say ``YES'' if $\lambda_1(\bL(\bB))\leq d$ and  ``NO'' if  $\lambda_1(\bL(\bB)) > \gamma d$.


\item[$\tsvp_{\gamma}$:] Given a lattice basis $\bB$, find a non-zero vector $\bv \in \bL(\bB)$ such that $\ltn{\bv} \leq \gamma\lambda_1(\bL(\bB))$.
%
\end{description}

\noindent
We now formally define the search and decision unique shortest vector problem. The definition of the decision version of $\tusvp$ is implicit in Cai \cite{Cai98}, although, to our knowledge, it has not been explicitly defined anywhere in the literature.  
\begin{description}
\item[$\tusvp_{\gamma}$:] Given a lattice basis $\bB$ such that $\lambda_2(\bL(\bB)) \geq \gamma\lambda_1(\bL(\bB))$, find a vector $\bv \in \bL(\bB)$ such that $\ltn{\bv} = \lambda_1(\bL(\bB))$.

\item[$\tdusvp_{\gamma}$:] Given a lattice basis $\bB$ and an integer $d$, such that $\lambda_2(\bL(\bB)) \geq \gamma\lambda_1(\bL(\bB))$, say ``YES'' if $\lambda_1(\bL(\bB))\leq d$ and  ``NO'' if $\lambda_1(\bL(\bB)) > d$.
\end{description}

\subsection{Defining \tcoam~ and \tconp~}

The definitions of this section have been adapted from \cite{GG98}.

\begin{definition}\em
A promise problem $\Pi = (\Pi_{\text{YES}}, \Pi_{\text{NO}})$ is said to be in \tconp~ if there exists a polynomial-time recognizable (witness) verification predicate $V$ such that 
\begin{itemize}
\item For every $x \in \Pi_{\text{NO}}$, there exists $w \in \{0,1\}^*$ such that $V(x, w) = 1$.
\item For every $x \in \Pi_{\text{YES}}$ and every $w \in \{0,1\}^*$, $V(x,w) = 0$.
\end{itemize}
\end{definition}

\begin{definition}\em
A promise problem $\Pi = (\Pi_{\text{YES}}, \Pi_{\text{NO}})$ is said to be in \tcoam~ if there exists a polynomial-time recognizable verification predicate $V$ and polynomials $p, q$ such that for every $x \in \Pi_{\text{YES}} \cup \Pi_{\text{NO}}$ with $|x| = n$, and $y$ chosen uniformly at random from $\{0,1\}^{p(n)}$,
\begin{itemize}
\item If $x \in \Pi_{\text{NO}}$, then there exists $w \in \{0,1\}^{q(n)}$, such that $\Pr \left(V(x, y, w) = 1 \right) \geq \frac{2}{3}$.
\item If $x \in \Pi_{\text{YES}}$, then for all $w \in \{0,1\}^{q(n)}$, $ \Pr \left( V(x, y, w) = 1 \right) \leq \frac{1}{3}$.
\end{itemize}
\end{definition}

\section{A deterministic polynomial time reduction from \tsvp~ to \tusvp~}

Let us suppose that $\bB=[\bb_1 ~ \bb_2 ~\dots ~ \bb_n]$ is the input lattice. 
The Gram Schmidt orthogonalization of $\bB$, denoted as $\{ \tilde{\bb}_1, \ldots, \tilde{\bb}_n\}$, is defined as
$$\tilde{\bb}_i = \bb_i - \displaystyle\sum\limits_{j  = 1}^{i-1}\mu_{i,j}\tilde{\bb}_j, \text{ where } \mu_{i,j} = \frac{\langle \bb_i, \tilde{\bb}_j \rangle}{\langle \tilde{\bb}_j, \tilde{\bb}_j \rangle} \; .$$

\begin{definition}\em
\label{def:LLL}
A basis $\bB = \{\bb_1, \ldots, \bb_n\}$ is a $\delta$-{\em LLL reduced basis} \cite{LLL82} if the following holds:
\begin{itemize}
\item $\forall \: 1 \leq j < i \leq n, \: |\mu_{i,j}| \leq \frac{1}{2}$,
\item $\forall \: 1 \leq i < n, \: \delta {\|\tilde{\bb}_i \|}^2 \leq \| \mu_{i+1, i}\tilde{\bb}_i + \tilde{\bb}_{i+1} \|^2$.
\end{itemize}
\end{definition}
We choose $\delta = \frac{3}{4}$ and then, from the above definition, for a $\delta$-LLL reduced basis,  $\forall \: 1 \leq i < n, \: \|\tilde{\bb}_{i}\| \leq \sqrt 2 \| \tilde{\bb}_{i+1} \|$. This implies that $$\|\tilde{\bb}_1\| \leq 2^{(i-1)/2} \| \tilde{\bb}_i\|\; .$$ Since there is an efficient algorithm \cite{LLL82} to compute an LLL-reduced basis, we assume, unless otherwise stated, that the given basis is always LLL-reduced and hence satisfies the above mentioned properties.

\begin{lemma}\label{lll}
For an LLL reduced basis $\bB$, if $\bu=\displaystyle\sum\limits_i \alpha_i \bb_i$ is a shortest vector, then $|\alpha_i| < 2^{3n/2}$ for all $i \in [n]$.
\end{lemma}
\begin{proof}
We show by induction that for $0 \leq i \leq n-1$, $|\alpha_{n-i}| \leq 2^{n/2+i}$. Since $\bu$ is the shortest vector of $\bL(\bB)$, $\| \bu \| \leq \| \bb_1\|$. Also, since the projection of $\bu$ in the direction of $\tilde{\bb}_n$ is $\alpha_n \tilde{\bb}_n$,
\begin{eqnarray*}
\| \tilde{\bb}_1 \| \geq \|\bu \| &\geq& |\alpha_{n}| \| \tilde{\bb}_n \| \\
         			  &\geq& 2^{-(n-1)/2} |\alpha_{n}| \| \tilde{\bb}_1 \| \; .
\end{eqnarray*}
This implies that $|\alpha_n| \leq 2^{(n-1)/2}$. 

Now assume that $|\alpha_{n-i}| \leq 2^{n/2+i}$ for $0 \leq i < k$. Then, using the fact that $\| \bu \| \leq \| \bb_1\|$ and that the projection of $\bu$ in the direction of $\tilde{\bb}_{n-k}$ is $\left( \alpha_{n-k} + (\displaystyle\sum\limits_{j = n-k+1}^n \mu_{j,n-k} \alpha_j) \right) \tilde{\bb}_{n-k}$, we get that 
\begin{eqnarray*}
\| \tilde{\bb}_1 \| \geq \|\bu \| &\geq& |\left( \alpha_{n-k} + (\displaystyle\sum\limits_{j = n-k+1}^n \mu_{j,n-k} \alpha_j) \right)| \|\tilde{\bb}_{n-k} \| \\
         			  &\geq& 2^{-(n-k-1)/2} |\left( \alpha_{n-k} + (\displaystyle\sum\limits_{j = n-k+1}^n \mu_{j,n-k} \alpha_j) \right)| \| \tilde{\bb}_1 \| \; .
\end{eqnarray*}
Therefore,
\begin{eqnarray*}
|\alpha_{n-k}| &\leq& 2^{(n-k-1)/2} + \displaystyle\sum\limits_{j = n-k+1}^n |\mu_{j,n-k} \alpha_j| \\
           &\leq& 2^{(n - k-1)/2} + \displaystyle\sum\limits_{j = 0}^{k-1} \frac{1}{2} |\alpha_{n-j}| \\
	   &\leq& 2^{(n - k-1)/2} + \frac{1}{2} \displaystyle\sum\limits_{j = 0}^{k-1}  2^{n/2 + j} \\
           &\leq& 2^{(n - k-1)/2} + \frac{1}{2}   2^{n/2 + k} \leq 2^{n/2 + k} \; .
\end{eqnarray*}
\qed \end{proof}

\subsection{Deterministic reduction from $\tsvp$ to $\tusvp$}

Given an instance of $\tsvp(\bB,d)$, we define a new lattice $\bL(\bB^{'})$ as follows.

\begin{equation*}
\left(
\begin{array}{cccc}
\bb_1 & \bb_2 & \dots & \bb_n\\ 
\frac{1}{2^{2n^2}} & 0 & \dots & 0\\
0 & \frac{2^{2n}}{2^{2n^2}} & \dots & 0\\ 
\vdots & \vdots & \vdots & \vdots \\
0 & 0 & \dots & \frac{2^{2n^2 -2n}}{2^{2n^2}} \\ 
\end{array}
\right)
\end{equation*}

So, $(\bb_i^{'})^{T}=[\bb_i^{T}~0~\dots~0~\dots~0~\frac{2^{2(i-1)n}}{2^{2n^2}}~0~\dots~0]$, where the $(m+i)$'th entry is non-zero. For a vector $\bv = \sum_i^n \alpha_i \bb_i \in \bL(\bB)$, we call $\bv' = \sum_i^n \alpha_i \bb_i'$ as the corresponding vector.

\begin{lemma}\label{lem:1}
For the new basis $\bB^{'}$, $\lambda_1^2(\bL(\bB)) \leq \lambda_1^2(\bL(\bB^{'})) \leq \lambda_1^2(\bL(\bB))+2^{-n/2}$.
\end{lemma}
\begin{proof}
The first inequality follows from the fact that the length of the vectors can't get shorter in $\bL(\bB^{'})$. For the second inequality, let $\bv$ be a shortest vector in $\bL = \bL(\bB)$ such that $\bv = \displaystyle\sum\limits_{i}^{n}\alpha_i \bb_i$. Then from Lemma \ref{lll}, $|\alpha_i| < 2^{3n/2}$, and hence
\begin{align*}
\ltsqn{\displaystyle\sum\limits_{i=1}^{n} \alpha_i \bb_i^{'}} &< \lambda_1^2(\bL) + \displaystyle\sum\limits_{i=0}^{n-1}\alpha_{i+1}^2\frac{2^{4in}}{2^{4n^2}}\\
&< \lambda_1^2(\bL) + 2^{3n}\frac{2^{4n^2} - 1}{(2^{4n} - 1)2^{4n^2}} \\
&< \lambda_1^2(\bL) + 2^{-n/2} \; .
\end{align*}
\qed \end{proof}

\begin{lemma}\label{lem:lambda1}
Let $\bv_1, \bv_2 \in \bL(\bB)$ be two distinct vectors such that $\ltn{\bv_1}=\ltn{\bv_2}=\lambda_1(\bL(\bB))$ and let $\bv_1^{'}, \bv_2^{'} \in \bL(\bB^{'})$ be the corresponding vectors. Then, $|\ltsqn{\bv_1^{'}}-\ltsqn{\bv_2^{'}}| > {2^{-4n^2}}$
\end{lemma}
\begin{proof}
Let $\bv_1 = \displaystyle\sum\limits_{i=1}^{n} \alpha_i \bb_i$ and $\bv_2=\displaystyle\sum\limits_{i=1}^{n} \beta_i \bb_i$. Let $j \in [n]$ be the largest number such that $\alpha_j \neq \beta_j$.  Then,

\begin{eqnarray*}
|\ltsqn{\bv_1^{'}} - \ltsqn{\bv_2^{'}}| &=& |\displaystyle\sum\limits_{i=1}^{n} (\alpha_i^2-\beta_i^2) (\frac{2^{2(i-1)n}}{2^{2n^2}})^2|\\
&>& |(\alpha_j^2 - \beta_j^2)\cdot \frac{2^{4(j-1)n}}{2^{4n^2}} + \displaystyle\sum\limits_{i=1}^{j-1} (\alpha_i^2 - \beta_i^2)\cdot \frac{2^{4(i-1)n}}{2^{4n^2}}| \\
&>& \frac{2^{4(j-1)n}}{2^{4n^2}} -2^{3n} \displaystyle\sum\limits_{i=1}^{j-1} \frac{2^{4(i-1)n}}{2^{4n^2}} \\
&=& \frac{2^{4(j-1)n}}{2^{4n^2}} -2^{3n}  \frac{2^{4(j-1)n} -1}{2^{4n^2} (2^{4n} - 1)} \\
&>& \frac{1}{2^{4n^2}} \; .
\end{eqnarray*}
\qed \end{proof}

\begin{lemma}\label{lem:3}
Let $\bv, \bv_1, \bv_2$ be vectors in an integer lattice $\bL = \bL(\bB)$. 
\begin{itemize}
\item If $\ltn{\bv_1} > \ltn{\bv_2}$, then $\ltsqn{\bv_1}-\ltsqn{\bv_2} \geq 1$.
\item If $\ltn{\bv}>\lambda_1(\bL)$, then if $\bv^{'} \in \bL(\bB^{'})$ is the corresponding vector, then $\ltsqn{\bv^{'}} > \lambda_1^2(\bB)+1$.
\end{itemize}
\end{lemma}
\begin{proof}
The first item follows from the fact that for integer lattices the $\ell_2^2$ norm of a vector is also an integer. The second item follows from the fact that $\bv$ is not the shortest vector in $\bL(\bB)$ and $\ltsqn{\bv^{'}} > \ltsqn{\bv}$.
\qed \end{proof}

Without loss of generality, we can assume $\bL(\bB)$ to be an integer lattice, and hence, using the above lemma, we get the following result.

\begin{theorem}
Given a lattice $\bL = \bL(\bB)$, there is a deterministic polynomial reduction transforming it to another lattice $\bL' = \bL(\bB^{'})$ such that $\frac{\lambda_2(\bL')}{\lambda_1(\bL')} > \sqrt{1+\frac{1}{c \cdot 2^{4n^2}\lambda_1^2(\bL)}}$ for some $c \leq 1/4$. In particular, $\tdusvp$ is {\bf NP}-hard under randomized reductions. 
\end{theorem}
\begin{proof}
From Lemma \ref{lem:lambda1} and Lemma \ref{lem:3}, we have that $\lambda_2^2(\bL') - \lambda_1^2(\bL')>2^{-4n^2}$, which implies $\frac{\lambda_2(\bL')}{\lambda_1(\bL')} > \sqrt{1+\frac{1}{2^{4n^2}\lambda_1^2(\bL')}}$. From Lemma \ref{lem:1}, $\lambda_1^2(\bL') < \lambda_1^2(\bL)+\frac{1}{2^{n/2}}$, and hence $\frac{\lambda_2(\bL')}{\lambda_1(\bL')}$ is at least $1+\frac{c}{2^{4n^2}\lambda_1^2(\bL)}$, for some constant $c\leq \frac{1}{4}$.
\qed\end{proof}

We would like to point out that we assumed in Lemma \ref{lem:3} that the lattice $\bL$ is an integer lattice. Hence, $\lambda_1(\bL)$ can be $O(2^{cn}\cdot\texttt{input size})$ and hence, $\frac{\lambda_2(\bL^{'})}{\lambda_1(\bL^{'})}$ can be arbitrarily close to~1. The original Kumar-Sivakumar \cite{KS01} proof also suffers with the same problem. The idea there is to show that the number of lattice points in a ball centered at the origin and of radius $\sqrt{2}\lambda_1(\bL)$ is at most $2^{n}$. Then one can create a new lattice $\bL^{'}$ with a unique short vector $\bv$ with $\lambda_1(\bL)\leq ||\bv|| < \sqrt{2}\lambda_1(\bL)$. In the worst case, the ratio of $\lambda_2^2(\bL^{'})$ and $\lambda_1^2(\bL^{'})$ for the new lattice (assuming that the original lattice was integer lattice) can be as small as $\frac{2\lambda_1^2(\bL)}{2\lambda_1^2(\bL)-1}$, which is $(1+\frac{1}{2\lambda_1^2(\bL)})$. As $\lambda_1(\bL)$ is $O(2^{cn}\cdot\texttt{input size})$, we get $(1+1/exp)$ hardness of $\tusvp$ in both cases.

\subsection{Deterministic hardness of $\tusvp$ in $\ell_\infty$ norm}

In this section, we show that the $\tusvp$ problem is $\tnp$-hard in the $\ell_\infty$ norm. For simplicity of description, we assume that all norms in this section are $\ell_\infty$ norms. Also, as before, the lattice $\bL$ is an integer lattice. 

For the LLL reduced basis $\{\tilde{\bb_1}, \dots, \tilde{\bb_n}\}$, there is a constant $c$ such that $\ltn{\tilde{\bb_1}} \leq 2^{c(i-1)}\ltn{\tilde{\bb_i}}$, for all $i \in [n]$. An induction proof as in Lemma \ref{lll} gives the following corollary.

\begin{corollary}
If the basis $\bB$ is LLL reduced then for the shortest vector $\bu = \sum_i \alpha_i \bb_i$, one has that for all $i$, $|\alpha_i|<2^{(c+1)n}$, for some constant $c$. 
\end{corollary}

We use the following theorem by P. van Emde Boas  \cite{B81}.
 
 \begin{theorem}\label{dinur}
 The problem $\tsvp$ in $\ell_\infty$ norm is $\tnp$-hard.
 \end{theorem}

Now we prove the main result of this section.
 
 \begin{theorem}
 The problem $\tusvp$ in $\ell_\infty$ norm is NP-hard.
 \end{theorem}
 \begin{proof}
 We take the instance resulting from Theorem \ref{dinur} and make the shortest vector unique. Let $\eta=(c+1)n$, then for all $i \in [n]$, $|\alpha_i|<2^{\eta}$. Given the basis $\{\bb_1, \dots, \bb_n\}$, we perturb the basis slightly in the following way. The basis vector $\bb_i$ gets $\frac{2^{2(i-1)\eta}}{2^{2\eta^2}}$ added to each of its entries. For the new lattice $\bL^{'}$, we have the following easy to prove observations. The theorem follows from them.
 \begin{itemize}
 \item If $\bv = \sum_i \alpha_i \bb_i \in \bL$ is a shortest vector then the vector $\bv^{'}=\sum_i \alpha_i \bb_i^{'} \in \bL^{'}$. Also, $$\lambda_1(\bL^{'}) \leq \ltn{\bv^{'}} \leq \lambda_1(\bL)+\sum_{i=1}^n \alpha_i \frac{2^{2(i-1)\eta}}{2^{2\eta^2}} = \lambda_1(\bL)+2^{1-\eta}.$$ 
 \item Let $\bv_1, \bv_2 \in \bL$ and $\ltn{\bv_1} > \ltn{\bv_2}$, then $\ltn{\bv_1} - \ltn{\bv_2} \geq 1$, as $\bL$ is an integer lattice.
 \item Let $\bv = \sum_{i \in [n]} \alpha_i \bb_i \in \bL$ and let $\bb_{i,j}$ be the $j$'th entry of $\bb_i$. If $\bv^{'}$ is the vector corresponding to $\bv$ in $\bL^{'}$ and $\ltn{\bv^{'}} = |\sum_{i \in [n]} \alpha_i \bb_{i,j}^{'}|$, for some $j \in [m]$, then $\ltn{\bv} = |\sum_{i \in [n]} \alpha_i \bb_{i,j}|$ for the same $j$ . This follows from the fact that the $\sum_{i\in [n]}\alpha_i \bb_{i,j}$ for all $j$ is an integer, and hence will either be equal to $\ltn{\bv}$ or will be at most $\ltn{\bv}-1$.
 \item Let $\bv_1, \bv_2 \in \bL$ such that $\ltn{\bv_1} = \ltn{\bv_2} = \lambda_1(\bL)$ then $|\ltn{\bv_1^{'}}-\ltn{\bv_2^{'}}| > |\sum_{i}(\alpha_i - \beta_i)\frac{2^{2(i-1)\eta}}{2^{2\eta^2}}|$. Similarly, as in Lemma \ref{lem:lambda1}, we get that $|\ltn{\bv_1^{'}}-\ltn{\bv_2^{'}}| > 2^{-2\eta^2}$.
 \end{itemize}
 \qed \end{proof}

\section{Hardness of \tusvp~ within $1 + 1/n^c$}

The following is a result obtained by letting $\eta = \frac{1}{40}$, $p= 2$, and $k = 1$ in Theorem 3.1 and Theorem 5.1 of \cite{Khot05}. 

\begin{lemma}
\label{lem:Khot}
 For some fixed constants $c_1, c_2$, there exists a polynomial time reduction from a \tsat~ instance of size $n$ to an \tsvp~ instance $(\bB, d)$ where $\bB$ is a  $2N \times N$ integer matrix with $N \leq n^{c_2}$, and $d \leq n^{c_1}$ such that:
\begin{enumerate}
\item If the \tsat~ instance is a YES instance, then with probability at least $9/10$, there exists a non-zero $\bx \in \mathbb{Z}^N$, such that $\| \bx\| \leq d^3$ and $\|\bB \bx\| \leq \sqrt{\frac{7}{8} d}$.
\item If the \tsat~ instance is a NO instance, then with probability at least $9/10$, for any non-zero $\bx \in \mathbb{Z}^N$, $\|\bB \bx\| \geq \sqrt{ d}$.
\end{enumerate}
\end{lemma}

We state below lemma 4 from \cite{KS01}.

\begin{lemma}
\label{lem:KS}
Let $T \neq \emptyset$ be a finite set of size at most $2^{m}$, and let $T = T_0 \supseteq T_1 \supseteq \cdots \supseteq T_{2m}$ be a sequence of subsets of $T$ defined by a probabilistic process that satisfies the following three properties:
\begin{enumerate}
\item For all $k, \; 0 \leq k < 2m$, and all $x \in T$, $\Pr(x \in T_{k+1} | x \in T_k) = \frac{1}{2}$.
\item For all $x \in T$, $0 \leq k \leq \ell < 2m$, $\Pr(x \in T_{\ell+1}| x \in T_{\ell}, x \in T_k) = \Pr(x \in T_{\ell+1}| x \in T_{\ell})$.
\item For all $k, \; 0 \leq k < 2m$, and all $x,y \in T_k$, $x \neq y$, the events $``x \in T_{k+1}"$ and $``y \in T_{k+1}"$ are independent. 
\end{enumerate}
Then, with probability $\frac{2}{3} - 2^{-m}$, one of the $T_k$'s has exactly one element. 
\end{lemma}

The following result is a simpler version of Corollary 3 from \cite{KS01}.
\begin{lemma}
\label{lem:svcount}
Given any arbitrary lattice $\bL$ of rank $n$, the number of lattice points in $\bL$ of length $\lambda_1(\bL)$ is at most $2^{n+1}$.
\end{lemma}
\begin{proof}
Let $\bB = (\bb_1, \ldots, \bb_n)$ be the basis of $\bL$. We claim that for any two vectors $\bu \neq \pm \bv \in \bL$ of length $\lambda_1(\bL)$, where $\bu = \displaystyle\sum\limits_{i=1}^n \alpha_i \bb_i$ and $\bv = \displaystyle\sum\limits_{i=1}^n \beta_i \bb_i$, there exists an $i$ such that $\alpha_i \not\equiv \beta_i \pmod 2$. Note that this claim implies the desired result.

Assume, on the contrary, that there exist a $\bu = \displaystyle\sum\limits_{i=1}^n \alpha_i \bb_i$ and $\bv = \displaystyle\sum\limits_{i=1}^n \beta_i \bb_i$ such that $\| \bu\| = \|\bv\| = \lambda_1(\bL)$ and $\alpha_i \equiv \beta_i \pmod 2$ for all $i$. This implies that $\frac{\bu + \bv}{2} \in \bL$ and $\frac{\bu - \bv}{2} \in \bL$. Also, 
\begin{eqnarray*}
\|\frac{\bu + \bv}{2}\|^2 + \|\frac{\bu - \bv}{2}\|^2 &=& \frac{\|\bu\|^2 + \| \bv\|^2 + 2 \langle \bu, \bv \rangle}{4} + \frac{\|\bu\|^2 + \| \bv\|^2 - 2 \langle \bu, \bv \rangle}{4} \\
&=& \frac{\|\bu\|^2 + \| \bv\|^2}{2} = \left(\lambda_1(\bL)\right)^2 \; .
\end{eqnarray*}
Since, $\bu \neq \pm \bv$, this implies that $0 < \|\frac{\bu + \bv}{2}\| < \lambda_1(\bL)$ and $0 < \|\frac{\bu - \bv}{2}\| < \lambda_1(\bL)$, which is a contradiction. 
\qed \end{proof}

We now prove the main result of this section.

\begin{theorem}
For some fixed constants $c_1, c_2, c$, there exists a polynomial time reduction from a \tsat~ instance of size $n$ to a sequence of lattice basis $\bB_i$, $1 \leq i \leq 2N+2$, and $d$, where $\bB_i$'s are $2N \times N$ integer matrices with $N \leq n^{c_2}$, and $d \leq n^{c_1}$ such that:
\begin{enumerate}
\item If the \tsat~ instance is a YES instance, then with probability at least $1/2$, there exists an $i$ such that $\bL(\bB_i)$ has a $1+\frac{1}{N^{c}}$-unique shortest vector of length at most $\sqrt{\frac{7}{8} d}$.
\item If the \tsat~ instance is a NO instance, then with probability at least $9/10$, for all $i$, the shortest vector of $\bL(\bB_i)$ is of length at least $\sqrt{ d}$.
\end{enumerate}
\end{theorem}
\begin{proof}
Given a \tsat~ instance, consider the pair $(\bB, d)$ using the reduction from Lemma \ref{lem:Khot}. 

We generate, as in \cite{KS01}, a sequence of lattices $\bL(\bB_0), \bL(\bB_1), \ldots, \bL(\bB_{2N+2})$ inductively as follows. Suppose we have generated $\bL(\bB) = \bL(\bB_0), \bL(\bB_1), \ldots, \bL(\bB_{k})$ for some $0 \leq k < 2N+2$. We now show how to generate $\bB_{k+1}$. Let $\bB_k = (\bb_1, \ldots, \bb_N)$. Pick a subset $W \subseteq [N]$ uniformly at random from all subsets of $[N]$. If $W$ is empty, then let $\bB_{k+1} = \bB_k$. Otherwise, pick any $i$ from $W$. For $j \notin W$, let $\bb_j' = \bb_j$, and for $j \in W \setminus \{i\}$, let $\bb_j' = \bb_j - \bb_i$. Finally, let $\bb_i' = 2 \bb_i$ and $\bB_{k+1} = (\bb_1', \bb_2', \ldots, \bb_N')$. 

Note that each of the $\bB_i$'s are $2N \times N$ integer matrices. Also, since $\bL(\bB_i) \subseteq \bL(\bB)$ for all $0 \leq i \leq 2N+2$, therefore, if the \tsat~ instance is a NO instance, then, by Lemma \ref{lem:Khot}, with probability $9/10$, the shortest vector of $\bL(\bB_i)$ is of length at least $\sqrt{ d}$ for all $i$. 

Now, consider the case when the \tsat~ instance is a YES instance. In this case, by Lemma \ref{lem:Khot}, with probability $9/10$, we have $1 \leq \lambda_1(\bL(\bB)) \leq \sqrt{\frac{7}{8} d}$, since, $\bB$ is an integer matrix. The set $T$ is a subset of $\bL(\bB)$ defined as follows:
\begin{equation*}
T = \{ \bv \in \bL(\bB) \;| \; \| \bv\| = \lambda_1(\bL(\bB)) \} \; .
\end{equation*}
Furthermore, we define the sets $T_i$ for $1 \leq i \leq 2N+2$ as $T_i = T \cap \bL(\bB_i)$. By Lemma \ref{lem:svcount}, $|T| \leq 2^{N+1}$. The sets $T_i$, for $1 \leq i \leq 2N+2$ satisfy the conditions of Lemma \ref{lem:KS} for $m = N+1$. Thus, by Lemma \ref{lem:KS}, with probability $\frac{2}{3} - 2^{-N-1}$, there exists a $0 \leq k \leq 2N+2$ such that $|T_k| = 1$. Note that $B_i$ is an integer matrix for all $i$. Thus, since $|T \cap \bL(\bB_k)| = 1$, we see that 
$$
\lambda_2(\bL(\bB_k)) \geq \lambda_1(\bL(\bB_k)) + 1 \geq \lambda_1(\bL(\bB_k)) (1 + \sqrt{\frac{8}{7d}}) \; .
$$ 
Thus, there exists a constant $c$ (which can be computed in terms of $c_1$ and $c_2$) such that with probability $\frac{9}{10} \cdot (\frac{2}{3} - 2^{-N-1}) > \frac{1}{2}$, there exists a $k$ such that $\bL(\bB_k)$ has a $(1+\frac{1}{N^{c}})$-unique shortest vector of length at most $\sqrt{\frac{7}{8} d}$. This concludes the proof. 
\qed \end{proof}

\section{From $\tgsvp \in \tconp$ ($\tcoam$) to $\tdusvp \in \tconp$ ($\tcoam$)}

We now simplify and generalize the $\tusvp_{n^{1/4}} \in \tcoam$ proof by Cai \cite{Cai98}. We first give a simplified description of Cai's proof that uses the idea of the $\tcoam$ proof of \cite{GG98}. Here, one needs to give a $\tcoam$ proof that given a lattice $\bL$ with $n^{1/4}$-unique shortest vector and an integer $d$, $\lambda_1(\bL) > d$. The protocol is as follows. The verifier generates uniform random points $\bp_i \in \bL$ for $i \in \{0,1,\dots,\log_2 (\min_i ||\bb_i||)\}$. For each $i$ the verifier generates a random point $\bz_i \in B(\bp_i, 2^{i-1}t\sqrt{\sqrt{n}-\frac{1}{4}})$. The verifier then sends these points to the prover. The prover then provides the claimed shortest vector $\bv$ (primitive vector) and for the correct range when $2^it<||\bv||\leq 2^{i+1}t$, the correct point $\bp_i\pmod \bv$ which is in $\bL$. If $\lambda_1(\bL) > d$ then the prover can send the correct shortest vector $\bv$ and for the corresponding $i$ the balls corresponding to different choices of $\bp \in \bL$ are disjoint or identical depending on whether the respective centers are congruent modulo the shortest vector $\bv$. So, the prover has no trouble in providing the proof when $\lambda_1(\bL) > d$. If on the other hand $\lambda_1(\bL) \leq d$ and $||v|| > d$, it must be a multiple of the shortest vector or much longer than $\lambda_1(\bL)$. In this case, the balls have lot of overlap and the prover will be caught with high probability.

We show that the above idea can be generalized for any $\tconp$ or $\tcoam$ proof, i.e., we show that for {\em any} factor $\gamma$, if $\tgsvp_\gamma \in \tconp$ then $\tdusvp_{c\sqrt{\gamma}}$ is in $\tconp$ (and similarly for $\tcoam$). This implies, using the result of Aharonov and Regev \cite{AR05} that $\tgsvp_{\sqrt{n}} \in \tconp$, that $\tdusvp_{cn^\frac{1}{4}} \in \tconp$, and any subsequent improvements in the factor for $\tgsvp$ will imply an improvement for $\tdusvp$.

\begin{lemma}
\label{lem:conp}
Let $\bL$ be a lattice such that $\lambda_2(\bL) \geq \gamma\lambda_1(\bL)$, and let $\bv$ be a primitive vector in $\bL$. Then:
\begin{itemize}
\item If $\|\bv\| \neq \lambda_1(\bL)$, then $\lambda_1(\bL_{\bot \bv}) \leq \frac{\|\bv \|}{\gamma} $.
\item If $\|\bv\| = \lambda_1(\bL)$, then $\lambda_1(\bL_{\bot \bv}) \geq \left(\sqrt{ \gamma^2 - \frac{1}{4}}\right) \| \bv\|$.
\end{itemize}
\end{lemma}
\begin{proof}
\begin{description}
If $\|\bv\| \neq \lambda_1(\bL)$ and $\bv$ is primitive, then $\|\bv \| \geq \lambda_2(\bL) \geq \gamma \lambda_1(\bL)$. Let $\bu$ be the shortest vector in $\bL$. Then the projection of $\bu$ in the space orthogonal to $\bv$ (say $\bu' \in \bL_{\bot \bv}$) is of length at most $\|\bu\| = \lambda_1(\bL)$. Also, $\bu$ is not parallel to $\bv$, and hence, $\bu' \neq \mathbf{0}$. This implies $$\lambda_1(\bL_{\bot \bv}) \leq \lambda_1(\bL)  \leq \frac{\|\bv \|}{\gamma} \; .$$
 
If $\|\bv\| = \lambda_1(\bL)$, then let $\bu'$   be the shortest vector in $\bL_{\bot \bv}$. Let $\bu'$ be the projection of $\bu \in \bL$ orthogonal to $\bv$. Then $\bu = \bu' + \alpha \bv$ for some $\alpha \in \R$. Since $\bu - \lfloor \alpha \rceil \bv \in \bL$ is not an integer multiple of $\bv$, $\| \bu - \lfloor \alpha \rceil \bv \| \geq \lambda_2(\bL) \geq \gamma \| \bv\|$. Thus, $$\gamma \| \bv\| \leq \|\bu' + (\alpha - \lfloor \alpha \rceil) \bv \| \leq \sqrt{ \|\bu' \|^2 + \frac{1}{4} \| \bv \|^2}\; ,$$ because $\bu'$ is orthogonal to $\bv$. This implies that 
$$\lambda_1(\bL_{\bot \bv}) = \| \bu'\| \geq \left(\sqrt{ \gamma^2 - \frac{1}{4}}\right) \| \bv\|\; .$$
\end{description} \ 
\qed \end{proof}

\begin{theorem}
\label{thm:conp}
If $\tgsvp_{\gamma \sqrt{\gamma^2 - \frac{1}{4}}} \in \tconp$, then $\tdusvp_{\gamma} \in \tconp$.
\end{theorem}
\begin{proof}
Let $(\bB, d)$ be an instance of $\tdusvp_{\gamma}$. Assume a witness for recognizing $\lambda_1(\bL(\bB)) > d$ to be a vector $\bv$ and a string $w$. The verification predicate $V$ on input $(\bB, d, \bv, w)$ outputs $1$ if and only if $\bv$ is a primitive vector of $\bL = \bL(\bB)$, $\|\bv\| > d$, and the verification predicate $V'$ for proving $\tgsvp_{\gamma'} \in \tconp$, (where $\gamma' = \gamma \sqrt{\gamma^2 - \frac{1}{4}}$) on input $(\bB', \frac{\| \bv\|}{\gamma}, w)$ outputs $1$, where $\bB'$ is a basis for $\bL_{\bot \bv}$.
\begin{description}
\item[CASE 1:] $(\bB, d)$ is a ``NO" instance, i.e. $\lambda_1(\bL) > d$. 

In this case, let $\bv$ be the shortest vector in $\bL$, and $w$ is the witness output in the proof of $\tgsvp_{\gamma'} \in \tconp$ for input $(\bB', \frac{\| \bv\|}{\gamma})$. 

Since $\lambda_1(\bL) > d$, $\bv$ is a primitive vector of $\bL$ with length greater than $d$. Also, from Lemma \ref{lem:conp}, $\lambda_1(\bL_{\bot \bv}) \geq \left(\sqrt{ \gamma^2 - \frac{1}{4}}\right) \| \bv\| = \gamma' \frac{\| \bv\|}{\gamma}$. 

Thus, the verification predicate $V$ outputs $1$.
\item[CASE 2:] $(\bB, d)$ is a ``YES" instance, i.e. $\lambda_1(\bL) \leq d$.

In this case, let us assume that there exists a witness $\bv, w$ such that $V$ outputs $1$. 

Thus, $\bv$ is a primitive vector with $\| \bv \| > d$. This implies that $\|\bv\| \neq \lambda_1(\bL)$, and using Lemma \ref{lem:conp}, $\lambda_1(\bL_{\bot \bv}) \leq \frac{\|\bv \|}{\gamma} $. Therefore, $V'$, and hence $V$, output $0$, which is a contradiction.
\end{description}
\qed \end{proof} 
 
This result, along with the result of \cite{AR05} implies the following:
\begin{corollary}
There exists $c > 0$ such that $\tdusvp_{c n^{1/4}} \in \tnp \cap \tconp$.
\end{corollary}

Note that essentially the same idea as in Theorem \ref{thm:conp} can be used to show that 
\begin{theorem}
If $\tgsvp_{\gamma \sqrt{\gamma^2 - \frac{1}{4}}} \in \tcoam$, then $\tdusvp_{\gamma} \in \tcoam$.
\end{theorem}

Thus, using the result of \cite{GG98}, this implies the following:
\begin{corollary}
There exists $c > 0$ such that $\tdusvp_{c (\frac{n}{\log n})^{1/4}} \in \tnp \cap \tcoam$.
\end{corollary}

\section{A deterministic reduction from $\tusvp_{\gamma}$ to $\tdusvp_{\gamma/2}$}

The following lemma is taken from the $\tusvp$ to $\tgsvp$ reduction given in \cite{LM09}.
\begin{lemma}
\label{lem:dim_red}
Let $\bL = \bL_0$ be a lattice of rank $n \geq 2$ given by its basis vectors, and let $\bu$ be the shortest non-zero vector of $\bL$. If there exists an efficient algorithm that computes a basis for $\bL_{i+1}$, a sub-lattice of $\bL_i$ such that $\bL_{i+1} \neq \bL_i$ and  $\bu \in \bL_{i+1}$ for all $i \geq 0$, then there exists an efficient algorithm that computes a basis for a sublattice $\tilde{\bL}$ of $\bL$ of rank $n-1$ such that $\bu \in \tbL$.
\end{lemma}
\begin{proof}
Let $\bB$ be the given basis for $\bL$, let $\bS$ be a basis for the sublattice $\bL_{t}$ for some $t > n(n + \log_2 n)$, and let $\bD$ be the dual basis of $\bS$. Since $\bL_{i+1}$ is a sub-lattice of $\bL_i$ for all $i$, we have that $\text{det}(\bS) \geq 2^t \text{det}(\bB)$, which implies $\text{det}(\bD) \leq 1/\left(2^t \text{det}(\bB)\right)$. By Minkowski's bound \cite{Min53}, we have $\lambda_1(\bL(\bD))\leq \sqrt{n} \text{det}(\bD)^{1/n}$, which implies that using the LLL algorithm \cite{LLL82}, we can find a vector $\bv \in \bL(\bD)$	such that
$$ \|\bv \| \leq 2^n \lambda_1(\bL(\bB)) \leq \frac{2^n \sqrt{n}}{2^{t/n} \text{det}(\bB)^{1/n}} \; .$$  
Also, using Minkowski's bound, we have $\|\bu \| \leq \sqrt{n} \text{det}(\bB)^{1/n}$. This implies that $$|\langle \bu, \bv \rangle | \leq \|\bu \| \| \bv \| \leq n \cdot 2^{n - t/n} < 1\; . $$

But $\bu \in \bL(\bD)$ and $\bv \in \bL(\bS)$, and thus $|\langle \bu, \bv \rangle |$ is an integer, which implies $\langle \bu, \bv \rangle = 0$, i.e., $\bu$ is perpendicular to $\bv$.  Thus, by taking the projection of $\bL$ perpendicular to $\bv$, we get a lattice $\tilde{\bL}$ in rank $n-1$ such that $\bu \in \tbL$.
\qed \end{proof}

\begin{lemma}
\label{lem:dusvporacle}
Let $\gamma \geq 2$ and $\bL$ be a lattice such that $\lambda_2(\bL) \geq \gamma \lambda_1(\bL)$. Then, given any sublattice $\bL'$ of $\bL$ containing the shortest non-zero vector $\bu$ of $\bL$ and an oracle that solves $\tdusvp_{\gamma/2}$, there exists an algorithm that computes a sublattice $\bL'' (\neq \bL')$ of $\bL'$ such that $\bu \in \bL''$.
\end{lemma}
\begin{proof}
Using the $\tdusvp_{\gamma/2}$ oracle, we can estimate $\|\bu\|$ within a factor of $2$ using binary search. Thus, let $d$ be such that $d/2 < \| \bu \| \leq d$.

Let $\bB = (\bb_1, \bb_2, \ldots, \bb_n)$ be a basis for $\bL'$ and let $\bu = \alpha_1 \bb_1 + \cdots + \alpha_n \bb_n$ be the shortest vector of $\bL$ for some $\alpha_i \in \Z$. Note that since $\bL'$ is a sub-lattice of $\bL$, $\lambda_2(\bL') \geq \lambda_2(\bL)$. 

Consider three basis as follows: $$\bB_1 = (2 \bb_1, \bb_2, \bb_3, \ldots, \bb_n) \; ,$$ $$\bB_2 = (\bb_1, 2 \bb_2, \bb_3, \ldots, \bb_n) \;, $$ $$\bB_3 = (\bb_1 + \bb_2, 2 \bb_2, \bb_3, \ldots, \bb_n) \; .$$

It is easy to see that $2 \bu$ belongs to each of $\bL(\bB_1)$, $\bL(\bB_2)$, and $\bL(\bB_3)$. Also, since these are sub-lattices of $\bL(\bB)$, $\lambda_2(\bL(\bB_i)) \geq \lambda_2(\bL(\bB))$. This implies that $\lambda_2(\bL(\bB_i)) \geq \frac{\gamma}{2} \lambda_1(\bL(\bB_i))$ for $i \in \{1, 2, 3\}$. 
Thus, using the $\tdusvp_{\gamma/2}$ oracle, we can check whether $\lambda_1(\bL(\bB_i)) \leq d$, or $\lambda_1((\bL(\bB_i)) > d$, and hence whether $\bu \in \bL(\bB_i)$ or not. 

It is sufficient to prove that $\bu \in \bL(\bB_i)$ for some $i \in \{1, 2, 3\}$. If $\alpha_1$ is even, then $\bu \in \bL(\bB_1)$, and if $\alpha_2$ is even, then $\bu \in \bL(\bB_2)$. If $\alpha_1$ and $\alpha_2$ are both odd, then $\bu = \alpha_1 (\bb_1 + \bb_2) + \frac{\alpha_2 - \alpha_1}{2}(2 \bb_2) + \alpha_3 \bb_3 + \cdots + \alpha_n \bb_n \in \bL(\bB_3)$. 
\qed \end{proof}

Thus, given a $\tusvp_{\gamma}$ instance $\bL(\bB)$ of rank $n$, using Lemma \ref{lem:dusvporacle}, we can obtain a sequence of sub-lattices (where each lattice is a strict sub-lattice of the previous one) such that each of these contains the shortest vector of $\bL(\bB)$. Then, using Lemma \ref{lem:dim_red}, we obtain a basis of a sublattice of $\bL(\bB)$ of rank $n-1$, still containing the shortest vector of $\bL(\bB)$. Repeating this procedure, we obtain a basis of a sublattice of $\bL(\bB)$ of rank $1$ containing the shortest vector of $\bL(\bB)$, which will be the vector $\bu$. We thus obtain the following result. 

\begin{theorem}
For any $\gamma \geq 2$, there exists an algorithm that solves $\tusvp_{\gamma}$ given a $\tdusvp_{\gamma/2}$ oracle. 
\end{theorem}

\section{Discussion and open problems}

Many interesting problems related to $\tusvp$ remain. The gap between the uniqueness factor $(1+\frac{1}{\tpoly})$, for which we know that the $\tusvp$ is hard, and $(\frac{n}{\log n})^{1/4}$, for which we know that the problem is in \tcoam~ is still large. It will be interesting to try to show hardness of $\tusvp$ for some constant factor.

The decision version of $\tusvp$ was not known to be $\tnp$-hard, as it does not follow from Kumar-Sivakumar's work \cite{KS01}. Our deterministic reduction from $\tsvp$ succeeds in showing the $\tnp$-hardness of the decision version but this hardness cannot be concluded even for a factor of $(1+\frac{1}{\tpoly})$ hardness, which remains an open problem. The search to decision equivalence of $\tdusvp$ and $\tusvp$ upto a factor of~2, shows that the complexity of the two problems is not too far apart. It is interesting to try to improve the factor of $2$, but this might require substantially new ideas. It is a major open question whether such a search to decision reduction is possible in the case of approximation versions of the shortest vector problem and the closest vector problem.

\end{document}